\documentclass[
]{ceurart}

\newif\ifarxiv\arxivtrue

\usepackage{listings}
\usepackage{amsthm}
\usepackage{stmaryrd}
\usepackage{xcolor}
\usepackage[ruled,vlined,linesnumbered,noresetcount]{algorithm2e}
\usepackage{caption}
\usepackage{subcaption}
\usepackage{bm}
\usepackage{float}
\usepackage{booktabs}
\usepackage{tabularx}
\usepackage{makecell}
\usepackage{hyperref}
\usepackage{changepage}
\usepackage{etoolbox}
\usepackage{verbatim}
\usepackage{amsfonts}

\usepackage{todonotes}

\usepackage{cleveref}
\usepackage{wasysym}

\usepackage{xspace}

\usepackage{wrapfig}
\usepackage{graphicx}
\usepackage{pgf}

\newtheorem{definition}{Definition}
\newtheorem{example}{Example}
\newtheorem{theorem}{Theorem}

\usepackage{tikz}
\usepackage{tkz-euclide}
\usepackage{tikz-3dplot}
\usepackage{tikz-qtree}
\usepackage{pgfplots} 
\usepgfplotslibrary{fillbetween}
\usetikzlibrary{arrows,arrows.meta,calc,shapes,shadows,decorations.pathmorphing,decorations.pathreplacing,automata,shapes.multipart,positioning,patterns,patterns.meta,external,backgrounds,fit,intersections}
\usepackage[outline]{contour}

\definecolor{rwth-blue}{cmyk}{1,.5,0,0}\colorlet{rwth-lblue}{rwth-blue!50}\colorlet{rwth-llblue}{rwth-blue!25}
\definecolor{rwth-violet}{cmyk}{.6,.6,0,0}\colorlet{rwth-lviolet}{rwth-violet!50}\colorlet{rwth-llviolet}{rwth-violet!25}
\definecolor{rwth-purple}{cmyk}{.7,1,.35,.15}\colorlet{rwth-lpurple}{rwth-purple!50}\colorlet{rwth-llpurple}{rwth-purple!25}
\definecolor{rwth-carmine}{cmyk}{.25,1,.7,.2}\colorlet{rwth-lcarmine}{rwth-carmine!50}\colorlet{rwth-llcarmine}{rwth-carmine!25}
\definecolor{rwth-red}{cmyk}{.15,1,1,0}\colorlet{rwth-lred}{rwth-red!50}\colorlet{rwth-llred}{rwth-red!25}
\definecolor{rwth-magenta}{cmyk}{0,1,.25,0}\colorlet{rwth-lmagenta}{rwth-magenta!50}\colorlet{rwth-llmagenta}{rwth-magenta!25}
\definecolor{rwth-orange}{cmyk}{0,.4,1,0}\colorlet{rwth-lorange}{rwth-orange!50}\colorlet{rwth-llorange}{rwth-orange!25}
\definecolor{rwth-yellow}{cmyk}{0,0,1,0}\colorlet{rwth-lyellow}{rwth-yellow!50}\colorlet{rwth-llyellow}{rwth-yellow!25}
\definecolor{rwth-grass}{cmyk}{.35,0,1,0}\colorlet{rwth-lgrass}{rwth-grass!50}\colorlet{rwth-llgrass}{rwth-grass!25}
\definecolor{rwth-green}{cmyk}{.7,0,1,0}\colorlet{rwth-lgreen}{rwth-green!50}\colorlet{rwth-llgreen}{rwth-green!25}
\definecolor{rwth-cyan}{cmyk}{1,0,.4,0}\colorlet{rwth-lcyan}{rwth-cyan!50}\colorlet{rwth-llcyan}{rwth-cyan!25}
\definecolor{rwth-teal}{cmyk}{1,.3,.5,.3}\colorlet{rwth-lteal}{rwth-teal!50}\colorlet{rwth-llteal}{rwth-teal!25}
\definecolor{rwth-gold}{cmyk}{.35,.46,.7,.35}
\definecolor{rwth-silver}{cmyk}{.39,.31,.32,.14}

\definecolor{chord-sat}{RGB}{128,133,233}
\definecolor{chord-unsat}{RGB}{241,92,128}
\definecolor{chord-memout}{RGB}{124,181,236}
\definecolor{chord-segfault}{RGB}{67,67,72}
\definecolor{chord-timeout}{RGB}{144,237,125}
\definecolor{chord-unknown}{RGB}{247,163,92}

\tikzset{
  term/.style={
              draw, 
              rectangle, 
              rectangle split, 
              rectangle split parts=2,
              rectangle split ignore empty parts=false, 
              rectangle split horizontal,
              append after command={%
		\pgfextra{\let\mainnode=\tikzlastnode 
                \coordinate (c1 \mainnode) at ($(\mainnode.west)!.5!(\mainnode.one split)$);
                \coordinate (c2 \mainnode) at ($(\mainnode.one split)!.5!(\mainnode.east)$);
		  }
                }}}

\makeatletter
\newcommand{\pgfplotsdrawaxis}{\pgfplots@draw@axis}
\makeatother
\pgfplotsset{only axis on top/.style={axis on top=false, after end axis/.code={
             \pgfplotsset{axis line style=opaque, ticklabel style=opaque, tick style=opaque,
                          grid=none}\pgfplotsdrawaxis}}}

\pgfplotsset{compat = newest}

\tikzdeclarepattern{
  name=mylines,
  parameters={
      \pgfkeysvalueof{/pgf/pattern keys/size},
      \pgfkeysvalueof{/pgf/pattern keys/angle},
      \pgfkeysvalueof{/pgf/pattern keys/line width},
  },
  bounding box={
    (0,-0.5*\pgfkeysvalueof{/pgf/pattern keys/line width}) and
    (\pgfkeysvalueof{/pgf/pattern keys/size},
     0.5*\pgfkeysvalueof{/pgf/pattern keys/line width})},
  tile size={(\pgfkeysvalueof{/pgf/pattern keys/size},
              \pgfkeysvalueof{/pgf/pattern keys/size})},
  tile transformation={rotate=\pgfkeysvalueof{/pgf/pattern keys/angle}},
  defaults={
    size/.initial=5pt,
    angle/.initial=45,
    line width/.initial=.4pt,
}, code={
\draw [line width=\pgfkeysvalueof{/pgf/pattern keys/line width}] (0,0) -- (\pgfkeysvalueof{/pgf/pattern keys/size},0);
}, }

\hypersetup{%
	pdftoolbar=false,
	pdfmenubar=false,
	colorlinks,
	urlcolor={black},
	linkcolor={rwth-blue},
	citecolor={black!50!rwth-green}
}

\RestyleAlgo{algoruled}
\SetKwComment{Comment}{// }{}
\SetAlgoSkip{smallskip}
\SetAlFnt{\footnotesize}
\SetKw{Continue}{continue}

\newcommand{\rat}{\texttt{SMT-RAT}}
\newcommand{\calc}{\texttt{CAlC}}
\newcommand{\calci}{\texttt{CAlC-I}}
\newcommand{\calcih}{\texttt{CAlC-IH}}

\newcommand{\True}{\texttt{True}}
\newcommand{\False}{\texttt{False}}

\newcommand{\QQ}{\ensuremath{{\mathbb{Q}}}}
\newcommand{\RR}{\ensuremath{{\mathbb{R}}}}
\newcommand{\N}{\ensuremath{{\mathbb{N}}}}
\newcommand{\NN}{\ensuremath{{\N}_{>0}}}

\newcommand{\I}{\ensuremath{\mathbb{I}}}

\renewcommand{\phi}{\ensuremath{\varphi}}
\renewcommand{\epsilon}{\ensuremath{\varepsilon}}
\renewcommand{\theta}{\ensuremath{\vartheta}}

\renewcommand{\l}{\lambda}

\renewcommand{\d}{\Delta}

\newcommand{\qfnra}{{\small QFNRA}\xspace}

\newcommand{\mathfont}[1]{\mathit{#1}}
\newcommand{\interior}[1]{\mathfont{int}(#1)}
\newcommand{\boundary}[1]{\mathfont{bound}(#1)}

\renewcommand{\approx}[1]{\underline{#1}}
\newcommand{\closure}{\mathfont{cl}}

\newcommand{\sgn}{\mathfont{sgn}}

\newcommand{\project}[2]{#1\!\downarrow_{#2}}
\newcommand{\proj}{\mathfont{proj}}

\newcommand{\sep}{\ensuremath{~\mid~}}

\newcommand{\res}{\ensuremath{\mathfont{res}}}

\newcommand{\disc}{\ensuremath{\mathfont{disc}}}

\newcommand{\ball}[2]{\ensuremath{B_{#1}(s)}}

\newcommand{\sample}{\textit{sample}}

\begin{document}

\copyrightyear{2023}
\copyrightclause{Copyright for this paper by its authors.
  Use permitted under Creative Commons License Attribution 4.0
  International (CC BY 4.0).}

\conference{SMT'23: $\text{21}^\text{st}$ International Workshop on Satisfiability Modulo Theories,
  July 05--06, 2023, Rome, Italy}

\title{Exploiting Strict Constraints in the\\Cylindrical Algebraic Covering}

\author[1]{Philipp Bär}[%
email=philipp.baer@rwth-aachen.de,
]
\address[1]{RWTH Aachen University, Germany}
\address[2]{United States Naval Academy, USA}

\author[1]{Jasper Nalbach}[%
email=nalbach@cs.rwth-aachen.de,
orcid=0000-0002-2641-1380
]
\cormark[1]

\author[1]{Erika {\'A}brah{\'a}m}[%
email=abraham@cs.rwth-aachen.de,
orcid=0000-0002-5647-6134
]

\author[2]{Christopher W. Brown}[%
email=wcbrown@usna.edu
]

\cortext[1]{Corresponding author.}

\begin{abstract}
  One of the few available complete methods for checking the satisfiability of sets of polynomial constraints over the reals is the  \emph{cylindrical algebraic covering (CAlC)} method. In this paper, we propose an extension for this method to exploit the \emph{strictness} of input constraints for reducing the computational effort. 
  We illustrate the concepts on a multidimensional example and provide experimental results to evaluate the usefulness of our proposed extension.
\end{abstract}


\maketitle

\section{Introduction}
\label{sec:introduction}

\emph{Quantifier-free (non-linear) real-algebraic (\qfnra)} formulas are Boolean combinations of polynomial constraints over the reals. Efficiently determining the \emph{satisfiability} of such formulas has become increasingly relevant in the last decades. A relatively recent general methodology, which has been proven very successful, encodes various properties of real-world systems by such formulas, e.g. the correctness of programs \cite{fm18,cbmc,software_mc,cimatti:smt-based} or security issues \cite{snelting}. To check whether the encoded properties hold or not, we need suitable software tools that can check the satisfiability of these encodings in a fully automated manner. 

Besides general-purpose computer algebra systems, this functionality is offered by some 
dedicated \emph{satisfiability modulo theories (SMT) solvers}. These tools use SAT solving to identify \emph{sets of polynomial constraints} whose truth satisfies the Boolean structure of an input formula, and require a \emph{theory solver} that can check the satisfiability of such constraint sets.

There are a few, even though not many, algorithms available that can be implemented to obtain such a theory solver. Incomplete methods like interval constraint propagation \cite{icp}, virtual substitution \cite{vs}, and subtropical satisfiability \cite{subtropical} are complemented by the complete method of the \emph{cylindrical algebraic decomposition (CAD)} \cite{collins} and a recent adaption of it named the \emph{cylindrical algebraic covering (CAlC)} \cite{abraham}. Employing CAlC as a theory solver in traditional SMT solving clearly improves over previous solving approaches regarding computational effort, as demonstrated by the cvc5 solver \cite{cvc5}, which has won the 2022 SMT competition \cite{smtcomp22} in the category of \qfnra using the CAlC algorithm.

The \qfnra satisfiability problem is decidable \cite{tarski} but hard. Both CAD and CAlC might need doubly exponential effort in the worst case \cite{DavenportHeintz1988} and there are currently no cheaper alternatives, even though \qfnra is known to be solvable in singly exponential time \cite{BrownDavenport2007}. Thus \emph{optimizations} play an important role to improve applicability and practical relevance.

In this paper we propose such an optimization for the CAlC algorithm: we exploit knowledge about the \emph{strictness of input constraints} in order to reduce the computational effort. For constraints in $n$ variables, the CAlC method iteratively computes subsets of $\RR^n$ over which the input constraint set is invariantly unsatisfiable, until either it finds a solution or the computed subsets cover $\RR^n$. Such a covering is composed of a finite number of open and non-open subsets, where each non-open one lies in the boundary of an open one. Previous work like e.g. \cite{brown} already exploited the well-known fact that if \emph{all} input constraints are strict then only the \emph{open} subsets need to be considered: if there is a solution in a non-open subset $N$ then there is also an $\epsilon$-ball of solutions around it, therefore also the open subset whose boundary contains $N$ will contain solutions. Skipping non-open subsets does not simply reduce the number of cases, but it avoids the hard cases, which typically require effortful computations with non-rational real-algebraic numbers. In this work, we go further and show that we can neglect some of the non-open subsets even if \emph{not all} (but some) of the input constraints are strict. This paper has three main contributions:
\begin{enumerate}
    \item We formalize the above-mentioned optimization for the CAlC method, to save effort in the presence of strict input constraints. Since a decomposition is a special type of covering, our results are transferable to the CAD method.
    \item We provide a publicly available implementation of the proposed optimization.
    \item We evaluate this implementation to evaluate and compare the modification to the original CAlC method.
\end{enumerate}

\noindent\emph{Outline:}\quad
We introduce the CAD and the CAlC methods in \Cref{sec:preliminaries} before we present our optimization in \Cref{sec:main}. We discuss the implementation in the \rat{} toolbox \cite{rat} and provide experimental results to demonstrate effectiveness in \Cref{sec:benchmarks}. Finally, we draw conclusions in \Cref{sec:conclusion}.

\section{Preliminaries}
\label{sec:preliminaries}

\paragraph{Quantifier-free non-linear real algebra (\qfnra)}
Let $\N$, $\NN$, $\QQ$ and $\RR$ denote the set of all natural, natural excluding zero, rational resp. real numbers. For $s \in \RR$, its \emph{sign} $\sgn(s)$ is $1$ if $s>0$, $0$ if $s=0$, and $-1$ otherwise. For a set $S$, we define $\mathcal{P}(S)=\{S'\,|\,S'\subseteq S\}$.

Assume for the rest of the paper $n\in\NN$ and some statically ordered \emph{variables} $x_1 {\prec} \dots {\prec} x_n$. Let $i\in\NN$ with $i\leq n$. By $\QQ[x_1, \dots, x_i]$ we denote the set of all \emph{polynomials} with variables $x_1, \dots, x_i$ and rational coefficients. The \emph{variety} of $p\in\QQ[x_1,\ldots,x_i]$ is the set of its \emph{real zeros} or \emph{roots} $\{s\in\RR^i\,|\,p(s)=0\}$.
A \emph{(\qfnra) constraint} has the form $p \sim 0$ with $p \in \QQ[x_1, \dots, x_n]$ and $\sim\, \in \{<,\leq,=,\neq,\geq,>\}$; $p \sim 0$ is \emph{strict} iff $\sim\;\in\{<,>,\neq\}$.
A \emph{(\qfnra) formula} $\varphi$ is a Boolean combination of constraints. By $\varphi(x_1,\ldots,x_i)$ we express that $\QQ[x_1,\ldots,x_i]$ includes all polynomials appearing in $\varphi$.
For $(s_1,\ldots,s_i) \in \RR^i$,
by $\varphi(s_1,\ldots,s_i,x_{i+1},\ldots,x_n)$ we denote the formula $\varphi[s_1/x_1]\ldots[s_i/x_i]$ after substituting $s_j$ for $x_j$ for $j=1,\ldots,i$.

A \emph{cell} is a non-empty, connected subset of $\RR^i$ for some $1\leq i\leq n$.
Let $S\subseteq\RR^i$ be a cell. 
Given $\varepsilon \in \RR$, $\varepsilon>0$ and $s \in \RR^i$, the \emph{$\varepsilon$-ball around $s$} is the set $\ball{\varepsilon}{s} := \{s'\in\RR^i \mid |s-s'|\leq \epsilon\}$.
$S$ is \emph{open} iff for all $s\in S$ there is an $\epsilon\in\RR$, $\epsilon>0$ such that $\ball{\varepsilon}{s}\subseteq S$. $S$ is \emph{closed} iff $\RR^i\setminus S$ is open.
The \emph{closure} $\closure(S)$ of $S$ is its smallest closed superset, its \emph{interior} $\interior{S}$ is its largest open subset, and its \emph{boundary} $\boundary{S}$ is its closure without its interior.

$S$ is \emph{semi-algebraic} if it is the solution set of a \qfnra formula.
Assuming $j\in\NN$ with $1\leq j<i$, the \emph{projection} of $S$ to $\RR^j$ is $\project{S}{j}=\{(s_1,\ldots,s_j)\,|\,(s_1,\ldots,s_i)\in S\}$; note that if $S$ is semi-algebraic then $\project{S}{j}$ is semi-algebraic. 

$S$ is \emph{sign-invariant} for $P\subseteq\QQ[x_1, \dots, x_i]$ if $\sgn(p(s))=\sgn(p(s'))$ for all $p\in P$ and all $s,s'\in S$.
$S$ is \emph{truth-invariant} (respectively \emph{UNSAT}) for a formula $\varphi(x_1,\ldots,x_i)$ if $\varphi(s)$ and $\varphi(s')$ simplify to the same truth value (respectively to \False) for all $s,s'\in S$.
For $P \subseteq \QQ[x_1,\ldots,x_i]$ and $s \in \RR^i$, by $S(P,s)$ we denote the maximal cell $S \subseteq \RR^i$ that is sign-invariant for $P$ and contains $s$.

Cell properties are generalized to sets of cells by requiring the respective property for each cell in the set. E.g., a cell set $D\subseteq\mathcal{P}(\RR^i)$ is semi-algebraic if each $S\in D$ is semi-algebraic. We also extend the projection to sets of cells $D\subseteq\mathcal{P}(\RR^i)$ as $\project{D}{j}=\{\project{S}{j}\,|\, S\in D\}$.

\paragraph{Cylindrical algebraic decomposition (CAD)}
\label{sec:cad}

The \emph{CAD method} \cite{collins} was introduced by Collins in 1975. Despite its doubly exponential complexity, all complete algorithms in SMT solvers for \qfnra  are based on the techniques underlying the CAD. This also holds true for the CAlC method.
For a formula $\varphi(x_1,\ldots,x_n)$, the CAD method decomposes $\RR^n$ into a finite set $D$ of cells that are sign-invariant for the polynomials in $\varphi$ and thus truth-invariant for $\varphi$, such that we can decide $\varphi$'s satisfiability by evaluating it at one sample point from each $S \in D$.


\begin{definition}[Cylindrical algebraic decomposition]
    Assume $i \in \NN$ with $i\leq n$.
    \begin{itemize}
        \item A \emph{decomposition} of $\RR^i$ is a finite set $D\subset\mathcal{P}(\RR^i)$  such that $\cup_{S \in D} S = \RR^i$ and either $S=S'$ or $S \cap S' = \emptyset$ for every $S,S' \in D$.
        \item A decomposition $D$ of $\RR^i$ is \emph{algebraic} if each $S\in D$ is a semi-algebraic cell.
        \item A decomposition $D_i$ of $\RR^i$ is \emph{cylindrical} if either $i=1$, or $i>1$ and $D_{i-1}=\{\project{S}{i-1}\mid S\in D\}$ is a cylindrical decomposition of $\RR^{i-1}$.
    \end{itemize}
\end{definition}

For a finite set of polynomials $P \subset \QQ[x_1,\ldots,x_n]$, the CAD method computes CADs $D_1, \ldots, D_n$ of $\RR^1,\ldots,\RR^n$ with $\project{D_i}{i{-}1}=D_{i-1}$ for $i=2,\ldots,n$, in two phases:\\ (1) The \emph{projection phase} computes for $i=n,\ldots,1$, in this order, finite sets of polynomials $P_i\subset \QQ[x_1,\ldots,x_i]$ such that the varieties of $\cup_{j=1}^i P_j$ define the boundaries of the cells in $D_i$.
Starting from $P_n=P$, for $i=n-1,\ldots,1$ we obtain $P_{i}$ by a \emph{projection operator} $\proj_i: \mathcal{P}(\QQ[x_1,\ldots,x_{i+1}]) \to \mathcal{P}(\QQ[x_1,\ldots,x_{i}])$ that eliminates the highest variable from its input polynomials. We do not define $\proj_i$ here, it suffices to assume that it maintains some properties such that the resulting CAD is sign-invariant for $P$.
\\
(2) In the lifting phase, compute the cells in the CADs $D_1, \ldots, D_n$, in this order. Instead of representing them explicitly, usually a sample point for every cell is generated. For $i\in\{1,\ldots,n\}$ and $s\in\RR^{i-1}$ (with $\RR^{0}=\{()\}$), let $\Xi_{s}=\{\xi_{s,1},\ldots,\xi_{s,k_s}\}$ be the ordered set of all real roots of the univariate polynomials from $\{p(s,x_i)\,|\,p\in P_i\}$. Let furthermore $I_s=\{(-\infty,\infty)\}$ if $\Xi_s=\emptyset$, and otherwise let $I_s$ consist of all point-sets  $\{\xi_{s,j}\}$, $j=1,\ldots,k_{s}$ (which cover the \emph{sections}, i.e. the cells consisting of roots of $P_i$), and the intervals $(-\infty,\xi_1),(\xi_1,\xi_2),\ldots,(\xi_{k_{s}-1},\xi_{k_{s}}),(\xi_{k_{s}},\infty)$ (which cover the \emph{sectors}, i.e. the open cells between the sections).
Assume furthermore a fixed but arbitrary function $\sample$ that assigns to each non-empty interval a value from this interval.
For $i=1,\ldots,n$, based on $S_{i-1}$ with $S_0=\{()\}$, we iteratively define the sample sets for $D_i$ as $S_i=\{(s,s_i)\,|\,s\in S_{i-1}\wedge I\in I_s\wedge s_i=\sample(I_s)\}$. 

\begin{example}
    \label{ex:cad}

    Consider $P=P_2 = \{ p_1: -x_1^2-x_2+1,\, p_2: x_1^2-x_2-1,\,  p_3: (x_1-0.5)^2+(x_2+1.5)^2-0.25,\, p_4: x_1+0.5 \}$ with projection $P_1 = \proj_2(P_2)$. \Cref{fig:ex-cad-calc} depicts the varieties of the polynomials from $P_2$, and below it those from $P_1$. The intervals defined by the real zeros from $P_1$ form a sign-invariant CAD $D_1$ for $P_1$. The sign-invariant CAD $D_2$ for $P_2$ is cylindrical over $D_1$, i.e. the cells of $D_2$ are arranged in cylinders over $D_1$. In $D_1$, the sections are the cells containing only a real root from $P_1$, the open intervals are the sectors. In $D_2$, we extend the varieties with the cylinder boundaries induced by $D_1$, and define $D_2$'s sections as the intersection points of lines as well as the line segments bounded by them, and the sectors as the open cells bounded by the sections.
    \Square
\end{example}

We note that not only the computation of the projection is computationally expensive, but also the number of samples grows exponentially with the number of variables. Although lifting over rational samples (that involves plugging in the sample point to obtain univariate polynomials and isolating their real roots) is rather efficient, in general, expensive lifting over non-rational \emph{real-algebraic} samples (which requires similar machinery as the projection) cannot be avoided. 

\begin{figure}[t]

    \begin{subfigure}[t]{0.49\textwidth}
        \centering
        \includegraphics{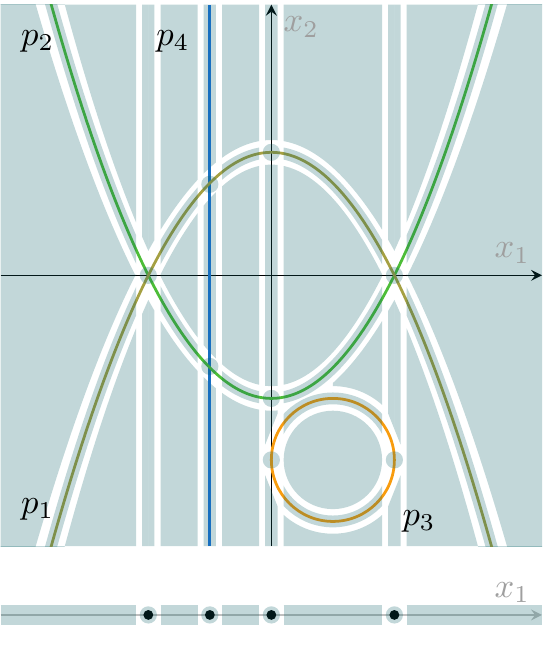}

        %
    \end{subfigure}\hfill
    \begin{subfigure}[t]{0.49\textwidth}
        \centering
        \includegraphics{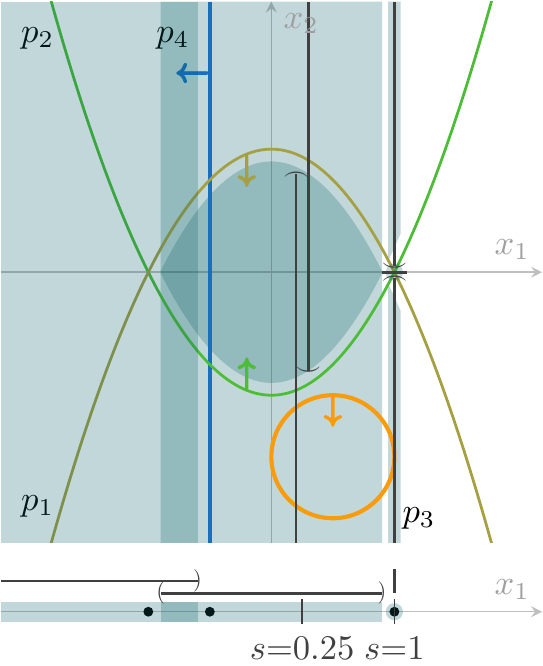}

    \end{subfigure}
    \caption{Examples for CAD (left) and CAlC (right). \emph{Left:} Varieties of $P_2$ (top) and $P_1$ (bottom) from Example \ref{ex:cad}. Cells of the sign-invariant CADs are coloured blue. \emph{Right:} Varieties of $P_2$ (top) and $P_1$ (bottom) from Example \ref{ex:calc}. Inwards-pointing arrows mark unsatisfiable areas for $\varphi$'s constraints (distinguished by different colours). UNSAT-intervals generated during the CAlC computations are shown as bracket-bounded thick black lines (open intervals) and short line segments (point intervals). UNSAT-cells for $\varphi$ are coloured blue. The one-dimensional samples $s=0.25$ and $s=1$ are shown as thin vertical lines.   }    
        \label{fig:ex-cad-calc}
\end{figure}

\paragraph{Cylindrical algebraic covering} \label{sec:calc}

CADs are finer than what we actually need for checking the satisfiability of \qfnra formulas, as we either need to find a satisfying sample, or \emph{cover} $\RR^n$ with UNSAT cells. \Cref{ex:cad} illustrates that cells which are UNSAT for a constraint are often split into several smaller cells due to sign changes of some other polynomials. To avoid such splits, the \emph{cylindrical algebraic covering (CAlC)} \cite{abraham} method, which uses the techniques from the CAD, relaxes sign-invariance for truth-invariance, still keeping the cylindrical arrangement of UNSAT cells but allowing overlaps between them.
Here we give a brief overview on the CAlC method for \emph{conjunctions} of constraints, and refer to \cite{abraham,qecovering} for the general case.

\begin{definition}[Cylindrical algebraic covering]
    Assume $i \in \NN$ with $i\leq n$.
    \begin{itemize}
        \item A \emph{covering} of $\RR^i$ is a finite set $C\subseteq\mathcal{P}(\RR^i)$ such that $\cup_{S \in C} S = \RR^i$.
        \item A covering $C$ of $\RR^i$ is \emph{algebraic} if each $S\in C$ is a semi-algebraic cell.
        \item A covering $C$ of $\RR^i$ is \emph{cylindrical} if either $i=1$, or $i>1$ and $C_{i-1}=\{\project{S}{i-1}\mid S\in C_i\}$ is a cylindrical covering of $\RR^{i-1}$.
    \end{itemize}
\end{definition}

The CAlC method is \emph{sample-guided}: in contrast to CAD, it does not start with projection but with a dimension-wise guess of values for a sample, with the aim to satisfy the input formula $\varphi$. For a current partial sample $(s_1,\ldots,s_{i-1})\in\RR^{i-1}$ with $1\leq i\leq n$, we iteratively identify intervals $I\subseteq \RR$ such that for any $s_i\in I$, $\varphi(s_1,\ldots,s_i)$ is unsatisfiable. We continue this process until either we find a solution (along with a partial UNSAT covering) or the intervals cover the whole $\RR$. In the latter case, if $i=1$ then the problem is unsatisfiable and a complete CAlC is returned; otherwise, we backtrack one level to sampling for $x_{i-1}$, generalize $s_{i-1}$ to an unsatisfying interval, and try to guess another value for $x_{i-1}$ outside the already excluded intervals. The interval generalization $I$ of $s_{i-1}$ is computed using real root isolation and a reduced version of the CAD projection. Intuitively,  it contains $s_{i-1}$ and other points $s_{i-1}'$ for which $\varphi(s_1,\ldots,s_{i-2},s_{i-1}')$ is unsatisfiable for same reason as $\varphi(s_1,\ldots,s_{i-2},s_{i-1})$.

In general, CAlCs require less computational effort than CADs: the sample-based search in CAlC saves parts of the projection executed in CAD, as well as the lifting effort corresponding to them.

\begin{example}
    \label{ex:calc}

    Re-using $P = \{ p_1: -x_1^2-x_2+1,\, p_2: x_1^2-x_2-1,\,  p_3: (x_1-0.5)^2+(x_2+1.5)^2-0.25,\, p_4: x_1+0.5 \}$ from \Cref{ex:cad}, we consider the formula $\varphi := p_1<0 \wedge p_2>0 \wedge p_3\geq0 \wedge p_4\geq0$. \Cref{fig:ex-cad-calc} illustrates on the right the following CAlC computations.

    Samples $x_1\in (-\infty, -0.5)$ violate $p_4\geq 0$. Outside $(-\infty, -0.5)$ we pick $s_1=0.25$ for $x_1$, and consider $\varphi(0.25,x_2)$.
    The constraint $p_1(0.25,x_2)<0$ is unsatisfiable for $x_2\in(-\infty,\frac{15}{16})$, and $p_2(0.25,x_2)>0$ is unsatisfiable for $x_2\in(-\frac{15}{16},\infty)$, together covering the real line.
    We generalize the UNSAT result to the cell $(-1,1)$ containing $s_1=0.25$ in the $D_1$ CAD for $\{p_1,p_2\}$ (see also \Cref{ex:strict}).

    Outside $(-\infty, -0.5)\cup(-1,1)$ we pick $s_1=1$ for $x_1$, and compute a covering of unsatisfying intervals for $\varphi(1,x_2)$ as depicted in the two-dimensional coordinate system. This time the sample $s_1=1$ is a real root, which cannot be generalized further than the section $[1,1]$ for $x_1$.

    Next, we pick a satisfying sample with $s_1=1.5$ and $s_2=0$, and the algorithm terminates.

    This example shows where the CAlC method is more efficient than the CAD: Firstly, single constraints like $p_4\geq0$ above can be used to rule out parts of the search space requiring fewer projection and lifting steps. Secondly, the CAlC method does not necessarily involve all polynomials in projection and corresponding lifting steps if they are redundant, as e.g. $p_3\geq0$ above excludes part of the search space that is already ruled out by the other constraints. 
    \Square
\end{example}

Each UNSAT interval $I \subseteq \RR$ generated during the CAlC computations over some partial sample $s \in \RR^{i-1}$ corresponds to an UNSAT cell $S \subseteq \RR^i$ with $\{s\}\times I=(\{s\} \times \RR) \cap S$. This UNSAT cell is represented \emph{implicitly} by a set of polynomials $P \subset \QQ[x_1,\ldots,x_i]$ and the sample $(s,s_i)$ for some $s_i\in I$ such that, except for some special cases, $S$ is the maximal sign-invariant cell for $P$ that contains $(s,s_i)$, i.e. $S = S(P,(s,s_i))$.

For generalizing a covering of intervals $I_1,\ldots,I_k \subseteq \RR$ (i.e. $\cup_{j=1}^k I_j = \RR$) over some sample point $s \in \RR^{i-1}$, the CAlC method defines a partial projection operator that projects the implicitly represented cells to an $(i-1)$-dimensional cell containing $s$ as follows:

\begin{definition}
    The \emph{covering projection operator} $\proj_{\textit{cov}}$ is a function which, for any $i,k\in\NN$, $P_1,\ldots,P_k \subset \QQ[x_1,\ldots,x_i]$, $s \in \RR^{i-1}$ and $s'_1,\ldots,s'_k \in \RR$ with $s \times \RR \subseteq \cup_{j=1}^k S(P_j,(s,s'_j))$ as input, returns a set of polynomials $P=\proj_{\textit{cov}}(P_1,\ldots,P_k,s,s'_1,\ldots,s'_k)\subseteq\QQ[x_1,\ldots,x_{i-1}]$ with the property that $S(P,s) \times \RR \subseteq \cup_{j=1}^k S(P_j,(s,s'_j))$.
\end{definition}

The correctness of an UNSAT covering is given by the following theorem:

\begin{theorem}
    \label{thm:calc}
    Let $\varphi$ be a formula in variables $x_1,\ldots,x_n$, $i>1$, $P_1,\ldots,P_k \subseteq \QQ[x_1,\ldots,x_i]$, $s \in \RR^{i-1}$ and $s'_1,\ldots,s'_k \in \RR$ such that $s \times \RR \subseteq \cup_{j=1}^k S(P_j,(s,s'_j))$ and for all $j=1,\ldots,k$ the cell $S(P_j,(s,s'_j))$ is UNSAT for $\varphi$.
    
    Then $S(\proj_{\textit{cov}}(P_1,\ldots,P_k,s,s'_1,\ldots,s'_k),s)$ is UNSAT for $\varphi$.
\end{theorem}

\section{Exploiting the Strictness of Constraints}
\label{sec:main}

Strict constraints $p<0$ are never satisfied at the real roots of $p$. Therefore, when checking the satisfiability of a formula $\varphi$ that contains only strict constraints (and no negations), sections can be neglected in CAD and CAlC computations, i.e. real roots can be omitted during sample construction.  This observation is exploited e.g. in \cite{brown}.

In this work we propose an approach that allows to omit real roots during sample construction in certain cases even if the input formula contains also weak constraints. We start with illustrating the idea on an example.

\begin{example}
  \label{ex:strict}
  In \Cref{ex:calc}, for each $s_1\in\RR$, both $p_1(s_1,x_2)$ and $p_2(s_1,x_2)$ have exactly one real root, which we denote by $\xi_1(s_1)$ respectively $\xi_2(s_1)$.

  For the sample $s_1=0.25$ for $x_1$, we covered $x_2$ by the intervals $(-\infty,\frac{15}{16})$ violating $p_1<0$, and  $(-\frac{15}{16},\infty)$ violating $p_2>0$. 
  The above intervals represent the cells $\{(s_1,s_2)\in\RR^ 2\,|\, s_2<\xi_1(s_1)\}$ and $\{(s_1,s_2)\in\RR^ 2\,|\, s_2>\xi_2(s_1)\}$. The generalization of $s_1=0.25$ is the maximal interval $(-1,1)$ over which the above cells still build a covering (see \Cref{fig:ex-cad-calc}).
  
  Note that $s_2=\xi_1(s_1)$ implies $p_1=0$, and $s_2=\xi_2(s_1)$ implies $p_2=0$. Thus, the \emph{closures} $\{(s_1,s_2)\in\RR^ 2\,|\, s_2\leq\xi_1(s_1)\}$ and $\{(s_1,s_2)\in\RR^ 2\,|\, s_2\geq\xi_2(s_1)\}$ are still unsatisfying for $p_1<0$ resp. $p_2>0$. Now, since both cells are closed, they cover $x_2$ over the closed generalization $[-1,1]$, the covering intervals at $s_1=\pm 1$ being $(-\infty,0]$ and $[0,\infty)$, which makes it unnecessary to consider the sample $s_1=1$ from \Cref{ex:calc}.
        \Square
\end{example}

Clearly, if we are able to deduce closed intervals, the computed CAlC consists not only of fewer cells, but we can also avoid computationally intensive lifting operations over potentially non-rational algebraic numbers.
The following example demonstrates that for two-dimensional formulas we could go even further.

\begin{example}
    \label{ex:strict-potential}
    Assume in the previous example that the first constraint would be non-strict, then for the constraints $p_1(0.25,x_2) \leq 0$ and $p_2(0.25,x_2)>0$ we would achieve the covering $(-\infty,\frac{15}{16})$ and $[-\frac{15}{16},\infty)$ for $x_2$. Now, only the second cell is closed. However, this still suffices to cover $[-1,1]$ also at its endpoints $s_1=\pm 1$ by $(-\infty,0)$ and $[0,\infty)$.
    \Square
\end{example}

\begin{figure}[t]
  \begin{subfigure}[b]{0.3\textwidth}
    \includegraphics[width=0.98\textwidth]{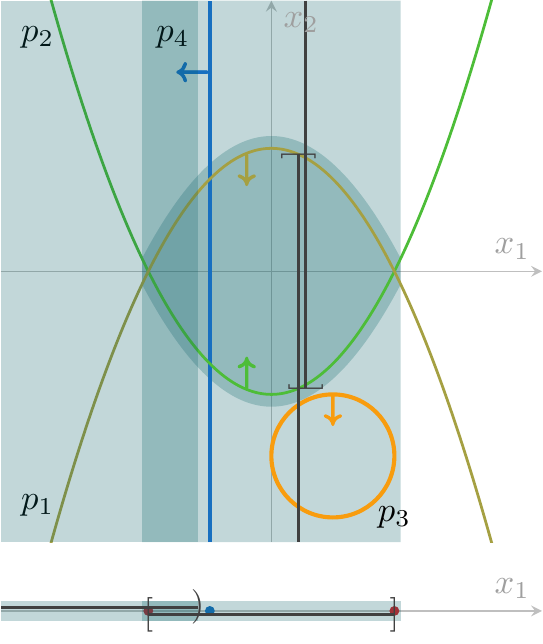}
    \caption{Closed cells from strict constraints.}
    \label{fig:ex-calc-strict}
  \end{subfigure}
  \hfill
  \begin{subfigure}[b]{0.3\textwidth}
    \includegraphics[width=0.98\textwidth]{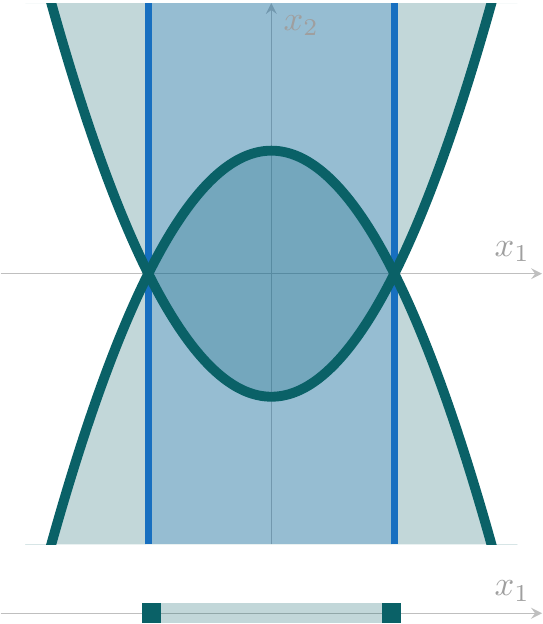}
    \caption{Closed generalizations from coverings of closed cells.}
    \label{fig:proof}
  \end{subfigure}
  \hfill
  \begin{subfigure}[b]{0.3\textwidth}
    \includegraphics[width=0.98\textwidth]{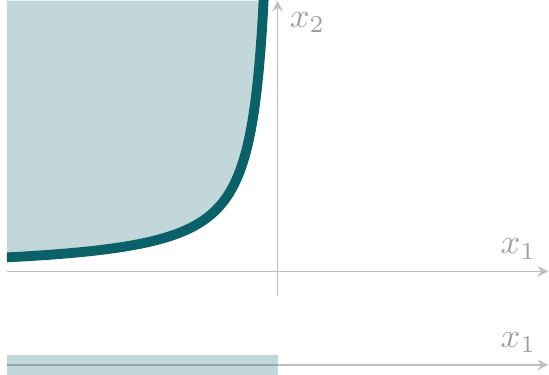}
    \caption{A closed set with open projection.}
    \label{fig:openproj}
  \end{subfigure}
  \caption{Projections of closed cells.}
\end{figure}%
However, generalizing this observation to constraint sets with more than two variables is non-trivial; we therefore focus on the case where \emph{all} intervals of a covering are closed.

To that end, we change our view from UNSAT intervals to the UNSAT cells they represent. During the CAlC computations, some UNSAT cells violate a constraint (as indicated by the arrows in \Cref{fig:ex-calc-strict}); as illustrated in \Cref{ex:strict}, we can \emph{close} those cells $S$ that violate \emph{strict} constraints $p\sim 0$ (i.e. $p(s)\not\sim 0$ for all $s\in S$) without losing the UNSAT property of the cell (i.e. $p(s)\not\sim 0$ for all $s\in \closure(S)$).

Furthermore, for any covering by UNSAT intervals which represent the $i$-dimensional \emph{closed} cells $S_1,\ldots,S_k \subseteq \RR^i$, if $S\subseteq \RR^{i-1}$ is a possible generalization of the current sample (i.e. $S \times \RR \subseteq \cup_{j=1}^k S_j$), then also the closure of $S$ is a valid generalization ($\closure(S) \times \RR \subseteq \cup_{j=1}^k S_j$). This fact is formalized in the following theorem and  visualized in \Cref{fig:proof}.


\begin{theorem} \label{thm:intervals}
    Assume $i,k\in\NN$, closed cells $S_1, \dots, S_k \subseteq \RR^i$, and a cell $S \subseteq \RR^{i-1}$ such that $S \times \RR \subseteq \cup_{j=1}^k S_j$. Then $\closure({S}) \times \RR \subseteq \cup_{j=1}^k S_j$.
\end{theorem}

\begin{proof}
  Assume for contradiction that there exists an $s \in \boundary{S} \times \RR\subseteq\RR^i$ such that $s \notin \cup_{j=1}^k S_j$. Then $s \notin S_j$, i.e. $s\in \RR^i\setminus S_j$ for all $j \in \{1,\ldots,k\}$. As the sets $S_j$ are closed, their complements $\RR^i \setminus S_j$ are open. By definition of an open cell, for each $j \in \{1,\ldots,k\}$ there exist $\varepsilon_j>0$ such that $\ball{\varepsilon_j}{s}\subseteq \RR^i \setminus S_j$. Let $\varepsilon$ be the smallest such $\varepsilon_j$ under all $j\in\{1,\ldots,k\}$. Then $\ball{\varepsilon}{s}\subseteq \RR^i \setminus S_j$ for all $j \in \{1,\ldots,k\}$.
    As $s \in \boundary{S} \times \RR=\boundary{S\times\RR}$, by the definition of $\ball{\varepsilon}{s}$ there exist an $s' \in \ball{\varepsilon}{s} \cap (S \times \RR)$, in other words, $s' \notin S_j$ for all $j \in {1,\ldots,k}$, which is a contradiction to the assumption that $S \times \RR \subseteq \cup_{j=1}^k S_j$.
\end{proof}

The above theorem might seem straight-forward at the first sight, but there are some special cases that make it less trivial, e.g. that the projection of a closed cell might not be closed as shown in \Cref{fig:openproj}. 

Now, we apply the general \Cref{thm:intervals} from above to \Cref{thm:calc} to obtain a variant that supports the derivation of closed cells.

\begin{theorem}
    \label{thm:calc-ext}

    Let $\varphi$ be a formula in variables $x_1\ldots,x_n$, $i>1$, $P_1,\ldots,P_k \subseteq \QQ[x_1,\ldots,x_i]$, $s \in \RR^{i-1}$ and  $s'_1,\ldots,s'_k \in \RR$ such that $s \times \RR \subseteq \cup_{j=1}^k \closure(S(P_j,(s,s'_j)))$ and for all $j=1,\ldots,k$ the cell  $\closure(S(P_j,(s,s'_j)))$ is UNSAT for $\varphi$.
    
    Then $\closure(S(\proj_{cov}(P_1,\ldots,P_k,s,s'_1,\ldots,s'_k),s))$ is UNSAT for $\varphi$.
\end{theorem}

\begin{proof}[Proof sketch]
    Let $l \geq k$, $P_{k+1},\ldots,P_l \subseteq \QQ[x_1,\ldots,x_i]$, $s'_{k+1},\ldots,s'_l \in \RR$ such that $s \times \RR \subseteq \cup_{j=1}^l S(P_j,(s,s'_j))$ and for every $j'=k+1,\ldots,l$ there exists a $j \in \{1,\ldots,k \}$ such that $P_{j'}=P_j$ and $(s,s'_{j'}) \in \boundary{S(P_j,(s,s'_j))}$, that means $S(P_{j'},(s,s'_{j'})) \subseteq \boundary{S(P_j,(s,s'_j))}$.

    Then by \Cref{thm:calc} it holds for $P := \proj_{cov}(P_1,\ldots,P_l,s,s'_1,\ldots,s'_l)$ that $S(P,s) \times \RR \subseteq \cup_{j=1}^l S(P_{j},(s,s'_{j})) = \cup_{j=1}^k \closure(S(P_{j},(s,s'_{j})))$.

    We now apply \Cref{thm:intervals} to $S(P,s)$ and $\closure(S(P_{j},(s,s'_{j})))$, $j=1,\ldots,k$ and obtain $\closure(S(P,s)) \times \RR \subseteq \cup_{j=1}^k \closure(S(P_{j},(s,s'_{j})))$.

    For the theorem, it remains to show that $P = \proj_{cov}(P_1,\ldots,P_k, s, s'_1,\ldots,s'_k)$.
    A formal proof would require the definition of the details of the projection operator $\proj_{cov}$, which we had to omit due to space restrictions. At this point, we just state without proving that adding the additional cells does not change the projection. We justify that as each additional cell $S(P_{j'},s'_{j'}),\ j=k+1,\ldots,l$ describes the boundary of a neighbouring open cell $S(P_j,s'_j)$ for some $j \in \{ 1,\ldots,k \}$ and the defining polynomial sets $P_j=P_{j'}$ are equal, and thus the `skeleton' of the covering does not change.
\end{proof}

\medskip

This result yields a simple adaption of the CAlC algorithm: Along with each implicit cell representation, we store a Boolean flag that indicates whether the corresponding unsatisfiable cell is closed or not. If the flag is set, the corresponding cell is closed, thus we can set the bounds of the witnessing interval to closed. Whenever we compute the base cell of a covering consisting of closed cells, we can easily deduce the flag for the new cell (it is true whenever the parents' flags are all true).
Thus, any implementation of the CAlC implementation can be adapted for the theorem by only superficial changes in the code.

\ifarxiv
We provide a 3D example for the covering method and our adaption in \Cref{sec:3dexample}.
\else
We provide a 3D example for the covering method and our adaption in the appendix of ??.
\fi

\section{Experimental results} \label{sec:benchmarks}

We implemented the proposed method to exploit strict constraints in our SMT-RAT \cite{rat} solver, using standard preprocessing and DPLL(T) solving with an implementation of the CAlC method

\begin{minipage}{\linewidth}
    \centering
    \hspace*{-0.55cm}
\begin{minipage}{0.4\textwidth}
    \centering
    \setkomafont{caption}{\sffamily\small}
    \setkomafont{captionlabel}{\usekomafont{caption}}
    \captionsetup{labelfont=bf,labelsep=newline,format=plain}
        \captionof{table}{Number of solved instances on the whole benchmark set.\\}
        \label{fig:results-overview}
        \begin{tabular}{m{1.2cm}m{1cm}m{1cm}}
            \toprule
            \small{Solver} & \small{SAT}    & \small{UNSAT}   \\ \midrule
            \calc  & $4553$ &  $4625$  \\
            \calci & $4610$ & $4648$   \\
            \calcih & $4609$ &  $4648$   \\ 
            \midrule
            Total & $5069$ & $5379$ \\
            & \multicolumn{2}{l}{($1104$ unknown)} \\ \bottomrule
        \end{tabular}
\end{minipage}%
\hspace*{0.5cm}%
\begin{minipage}{0.5\textwidth}
    \centering
        \setkomafont{caption}{\sffamily\small}
        \setkomafont{captionlabel}{\usekomafont{caption}}
        \captionsetup{labelfont=bf,labelsep=colon,format=plain}
        \includegraphics[width=0.99\textwidth]{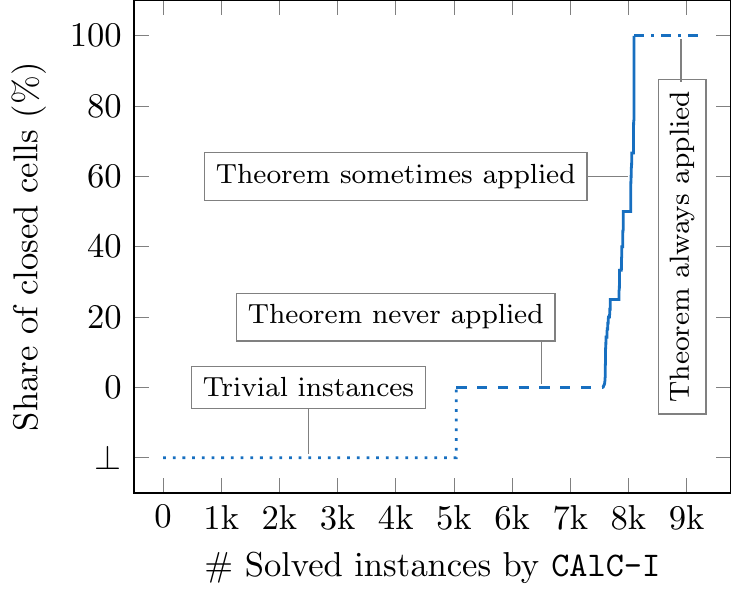}
        \captionof{figure}{Number of instances solved by \calci{} and their (maximal) ratio of flagged/closed cells.\\}
        \label{fig:theorem-applicable}
\end{minipage}
\end{minipage}

\noindent as the only theory solver. The implementation is accessible at \url{https://doi.org/10.5281/zenodo.7900518}. We execute the solver on the \textit{QF\_NRA} benchmark library from \texttt{SMT-LIB} \cite{smtlib} (as of April 2022), consisting of $11552$ instances that stem from $11$ different families. Each formula is solved on a CPU with $2.1$ GHz with a timeout of $60$ seconds and a memory limit of $4$ GB.
In the following, we denote the original CAlC solver by \calc{} and the modification that maintains interval flags to indicate closed cells by \calci{}.


\Cref{fig:results-overview} shows the overall performance of \calc{} and \calci.
The modified method \calci{} solves $80$ instances more than \calc{}, whereby more than two thirds of this gain is on SAT instances, the remaining on UNSAT instances.
\calc{} times out on $77$ instances which are solved by \calci{}. Conversely, \calci{} times out on $7$ instances which are solved by \calc{}. 

\paragraph{Applications of the theorem}

\Cref{fig:theorem-applicable} depicts the share of the \emph{derived} implicit cell representations that carry a \True{} flag (i.e. \Cref{thm:calc-ext} was applicable during its creation).
More than half of the instances are already solved by the SAT solver (the theory solver is never called) or the call to the CAlC contains only univariate polynomials. The theorem is never applied on $2518$ instances. The theorem is applied at least once in $1699$ instances. The theorem is always applied on $1162$ instances (i.e. all cells are closed).

We focus the further analysis on the interesting instances where the theorem was applied at least once: the instances solved by both solvers and the theorem was applied at least once ($1623$ instances), instances solved only by \calc{} ($8$ instances), and instances solved only by \calci{} ($88$ instances).
\begin{figure}[t]
    \centering
    \hspace*{-0.6cm}
    \begin{subfigure}[c]{0.47\textwidth}
        \centering
        
        \includegraphics[width=0.99\textwidth]{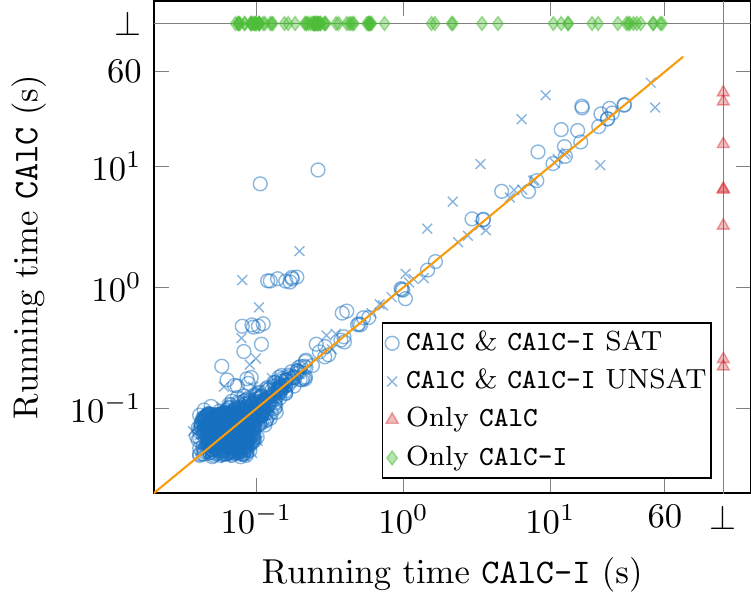}
        \caption{Running times of \calc{} and \calci{}.}
        \label{fig:calci-runtime}
    \end{subfigure}
    \hspace*{0.4cm}
    \begin{subfigure}[c]{0.47\textwidth}
        \centering
        
        \includegraphics[width=0.99\textwidth]{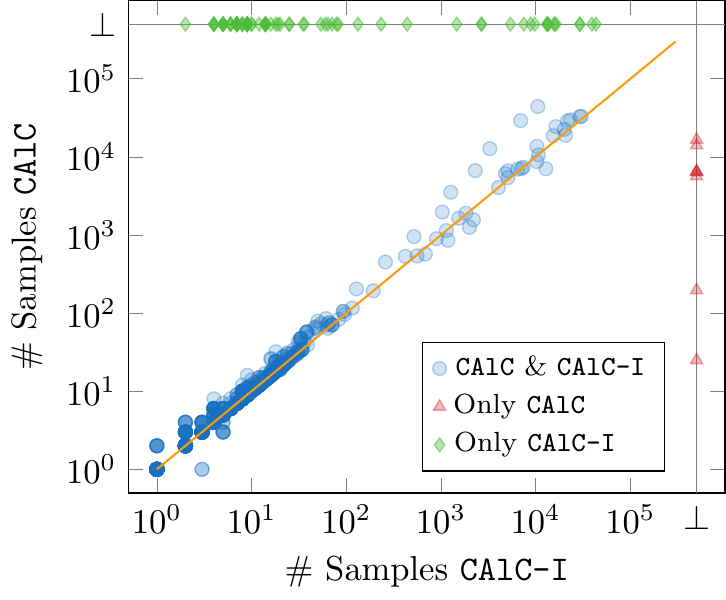}
        \caption{Samples of \calc{} and \calci{}}
        \label{fig:calci-samples}
    \end{subfigure}
    \caption{Scatter plots for running time and number of samples. $\bot$ denotes a timeout on an instance.}
\end{figure}%
\paragraph{Running times}

\Cref{fig:calci-runtime} compares the running times of \calc{} and \calci{} on the interesting instances. 
The bottom left cluster of instances is computationally easy; small deviations in running time between \calc{} and \calci{} are negligible. Most instances with higher running time are located near the equality line, i.e. \calc{} and \calci{} are equally efficient on them.
There are $30$ instances solved faster by more than $0.1$ seconds on \calc{} than on \calci{} (including instances solved only by \calc{}, see area below the equality line). The other way around, $197$ instances are solved faster by \calci{} (see area above the equality line). We conclude that \calci{} can significantly improve the running time on certain instances.

In particular, some UNSAT instances (blue crosses) deviate from the equality line for running times greater than one second. \calci{} is able to cover the entire space with fewer UNSAT cells than \calc{}; this effect is less significant on the SAT instances, where only a partial covering is computed.

\paragraph{Number of samples}

We expect that the advantage of the proposed optimization is due to the fact that the application of \Cref{thm:calc-ext} reduces the number of samples constructed in the CAlC method.

To evaluate this correlation, \Cref{fig:calci-samples} compares the number of partial sample points of \calc{} and \calci{}. Though the number of samples is often similar, \calci{} tends to generate fewer samples especially on the larger instances.

\paragraph{Iterative applications of the theorem}

As \Cref{thm:calc-ext} can only be applied if \emph{all} cells forming the covering are closed, the question raises how often it can be applied iteratively. To that end, we say that the cells forming a covering are the \emph{parents} of the covering's base cell. In that sense, we define for every cell its \emph{depth} as the distance to its `oldest' ancestor. 
\Cref{fig:calci-depth} plots for every instance the maximal depth among all cells versus the relative maximal depth among all closed cells.

We clearly see that the theorem is mostly applicable to cells with low depth. In particular, the $1308$ instances at the top left of the plot, where at least one closed cell has maximal depth, have a total maximal depth below $10$. Further, the instances where the theorem is only applied to cells with depth one form a visible hyperbola, which are 1159 of the 1711 depicted instances. We conclude that the performance gains stem mostly from `superficial' application of the theorem.
Still, there are some instances above this hyperbola, representing non-trivial applications of our theorem.

\begin{figure}[t]
    \centering
    \hspace*{-0.6cm}
    \begin{subfigure}[c]{0.47\textwidth}
        \centering
        
        \includegraphics[width=0.99\textwidth]{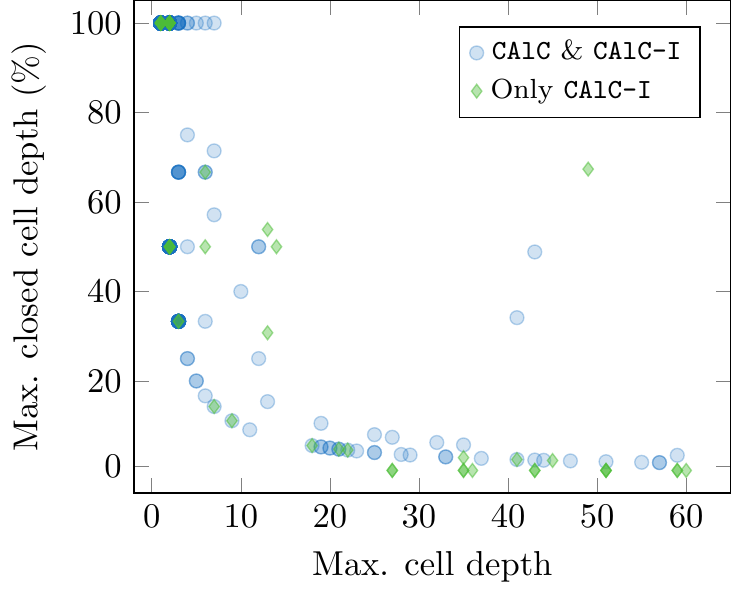}
        \caption{Cell depth ratio of \calci{}}
        \label{fig:calci-depth}
    \end{subfigure}
    \hspace*{0.4cm}
    \begin{subfigure}[c]{0.47\textwidth}
        \centering
        
        \includegraphics[width=0.99\textwidth]{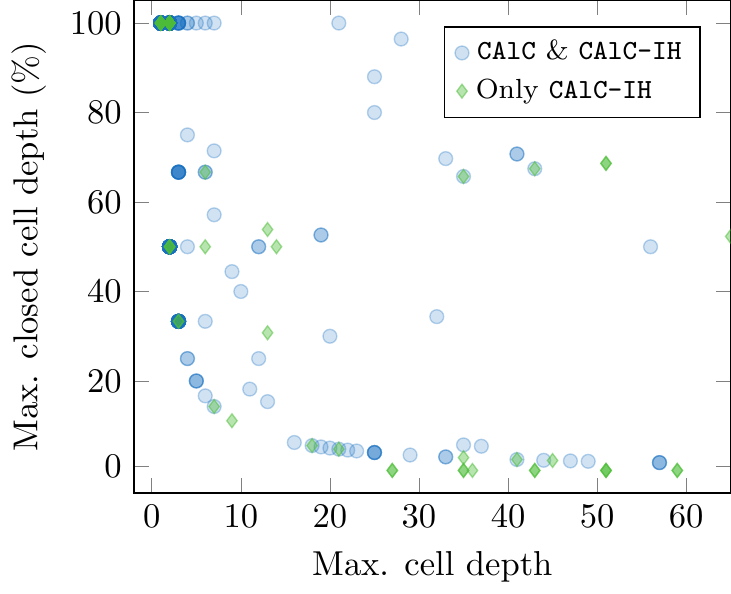}
        \caption{Cell depth ratio of \calcih{}}
        \label{fig:calcih-depth}
    \end{subfigure}
    \caption{Maximal cell depth and the relative maximal closed cell depth.}
\end{figure}%

\paragraph{Modification of the covering heuristic}

The modification \calci{} is clearly more efficient than \calc{}; as we just observed, these gains are due to `superficial' application of \Cref{thm:calc-ext}. We now aim to adapt the \calci{} method to support the application of the theorem also to cells of higher depth.
Whenever the CAlC method finds a covering of cells at some sample, the choice of the cells forming the covering is not unique; redundancies in the formula might allow for a choice of the cells. \calci{} heuristically minimizes the number of cells forming a covering. We propose an adaption \calcih{} which first tries to cover using closed cells only, and falls back to the default heuristic if it fails.  

\Cref{fig:calcih-depth} shows that the theorem is now applied on more instances with higher depth. However, the running times do not improve significantly (not shown here); \Cref{fig:results-overview} shows that even one instance less is solved. Thus, more sophisticated heuristics are desirable which find a better trade off between a covering of `good' intervals and supporting the theorem applicability.

\section{Conclusion} \label{sec:conclusion}

The cylindrical algebraic covering method admits reducing the number of projection and lifting operations compared to the cylindrical algebraic decomposition switching from being sign-invariant for a set of polynomials to being truth-invariant for an input formula. In this paper, we propose a natural extension to the CAlC method that exploits strict constraints in the input formula: If a strict constraint is unsatisfiable in some cell, then it is so in the closure of that cell. Our adaption allows to carry this information through the CAlC algorithm to avoid lifting over roots of polynomials, which is desirable as this is usually computationally expensive.

The proposed adaption is easy and efficient to implement. Our experimental evaluation concluded that (1) we gain a good portion of newly solved instances, (2) these gains are due to reduced number of lifting steps, and (3) our modification is still `superficial' and leaves potential for future investigation into the topic.

Future work consists of advanced theoretical work as motivated by \Cref{ex:strict-potential}, and better heuristics for choosing good coverings as motivated in the last paragraph of \Cref{sec:benchmarks}.

\begin{acknowledgments}
  Jasper Nalbach was supported by the DFG RTG 2236/2 \textit{UnRAVeL}.
  We thank James Davenport and Matthew England for fruitful discussions.
\end{acknowledgments}

\bibliography{bib}

\ifarxiv
\newpage
\appendix
\section{3D example}
\label{sec:3dexample}

We give a more elaborated 3D example to illustrate the main principles of the CAlC method and the proposed modification.

\paragraph{Constraint set}
Let $\varphi$ be the conjunction of constraints
\begin{displaymath}
    \varphi := c_1:z-y^2-x^2 \geq 0 \wedge c_2:z^2+y^2+x^2-3 < 0 \wedge c_3:z^2+(y-2.5)^2+x^2-16 > 0
\end{displaymath}

\noindent whose polynomial zeros are depicted in \Cref{fig:caclw}. We refer to the defining polynomial of $c_j$ as $p_j$. The zero $p_1=0$ is illustrated in blue, $p_2=0$ in orange, and $p_3=0$ in green. The constraint $c_1$ is satisfied inside the paraboloid and on its surface, $c_2$ inside the smaller sphere, and $c_3$ outside the larger one. The zeros of $p_2$ and $p_3$, i.e. the surfaces, are not satisfying. That way, every pair of constraints is satisfiable, but all together are not. A 3D-illustration is given at \url{https://doi.org/10.5281/zenodo.6738566}. 

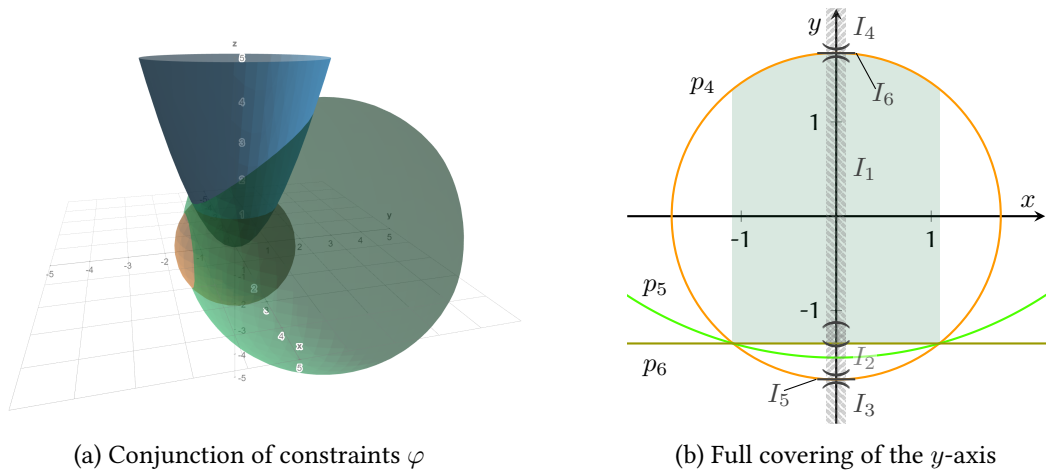
\begin{figure}[h]
    \centering
    \begin{subfigure}[b]{0.5\textwidth}
    \includegraphics[width=1\textwidth]{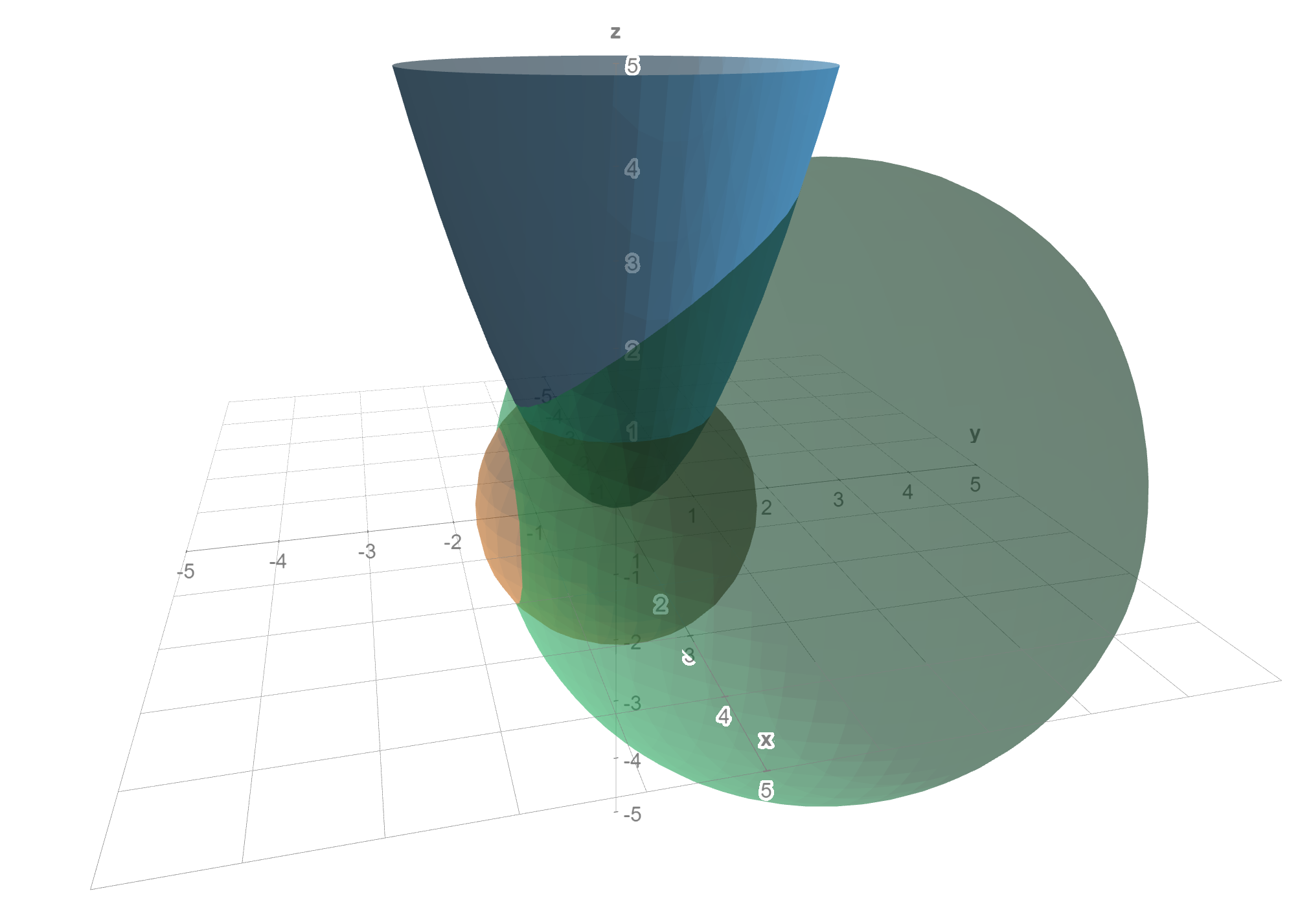}
    \caption{Conjunction of constraints $\varphi$}
    \label{fig:caclw}
    \end{subfigure}
    \hfill
    \begin{subfigure}[b]{0.45\textwidth}
        \centering

        \begin{tikzpicture} 
            \begin{axis}[axis on top=false,
                axis line style=thick, 
                axis x line=center, 
                axis y line=center, 
                ymin=-2.2,ymax=2.2,xmin=-2.2,xmax=2.2, 
                xlabel=$x$, ylabel=$y$,grid=none,
                xtick={-1,0,1},
                ytick={-1,0,1},
                yticklabels={-1,0,1},
                xticklabels={-1,0,1},
                y label style={at={(axis cs:-0.4,2.2)}},
                x=1.25cm,y=1.25cm
            ] 
                \begin{scope}
                    \clip[] (axis cs: -1.09,-1.35) rectangle (axis cs: 1.09,2);
                    \fill[color=rwth-teal, opacity=0.15] (axis cs:0,0) circle [radius=1.73205*1.25cm];
                \end{scope}
    
                \draw[thick, rwth-orange] (axis cs:0,0) circle [radius=1.73205*1.25cm];
                
                \node[label={$p_4$}] at (axis cs:-1.4,1.1) {};
    
                \addplot [domain = -5:5,
                        smooth,
                        samples=501,
                        thick,
                        color=rwth-green,
                        ]
                        {sqrt(16-x^2)+2.5};
                \addplot [domain = -5:5,
                        smooth,
                        samples=501,
                        thick,
                        color=rwth-green,
                        ]
                        {2.5-sqrt(16-x^2)};
                
                \node[label={$p_5$}] at (axis cs:-1.9,-1.05) {};
    
                \addplot [domain = -5:5,
                        smooth,
                        samples=200,
                        thick,
                        color=rwth-yellow!60!black,
                        ]
                        {-1.35};
                
                \node[label={$p_6$}] at (axis cs:-1.9,-1.9) {};
                
                \fill[pattern={mylines[size=2pt,line width=1.5pt,angle=-45]}, pattern color=gray, opacity=0.6] (axis cs:-0.1,-1.35) rectangle (axis cs:0.1,1.73);
                \fill[pattern={mylines[size=2pt,line width=1.5pt,angle=45]}, pattern color=gray, opacity=0.6] (axis cs:-0.1,1.73) rectangle (axis cs:0.1,2.2);
                \fill[pattern={mylines[size=2pt,line width=1.5pt,angle=45]}, pattern color=gray, opacity=0.6] (axis cs:-0.1,-1.73) rectangle (axis cs:0.1,-1.14);
                \fill[pattern={mylines[size=2pt,line width=1.5pt,angle=-45]}, pattern color=gray, opacity=0.6] (axis cs:-0.1,-2.2) rectangle (axis cs:0.1,-1.73);
    
                \node[color=darkgray] at (axis cs:0,1.73-0.05) {\rotatebox{90}{$\pmb{)}$}};
                \node[color=darkgray] at (axis cs:0,-1.35) {\rotatebox{90}{$\pmb{(}$}};
                \node[color=darkgray] at (axis cs:0,1.73+0.07) {\rotatebox{90}{$\pmb{(}$}};
                \node[color=darkgray] at (axis cs:0,-1.14) {\rotatebox{90}{$\pmb{)}$}};
                \node[color=darkgray] at (axis cs:0,-1.73+0.07) {\rotatebox{90}{$\pmb{(}$}};
                \node[color=darkgray] at (axis cs:0,-1.73-0.05) {\rotatebox{90}{$\pmb{)}$}};
                \draw[color=darkgray,thick] (axis cs:-0.2,1.73) -- (axis cs:0.2,1.73);
                \draw[color=darkgray,thick] (axis cs:-0.2,-1.73) -- (axis cs:0.2,-1.73);

                \node[color=darkgray] at (axis cs:0.3,0.5) {$I_1$};
                \node[color=darkgray] at (axis cs:0.3,2) {$I_4$};
                \node[color=black,fill=white,fill opacity=0.5,inner sep=0pt,outer sep=0pt] at (axis cs:0.3,-1.5) {$I_2$};
                \node[color=darkgray] at (axis cs:0.3,-2) {$I_3$};
                \node[color=darkgray] (name5) at (axis cs:0.5,1.3) {$I_6$};
                \node[color=darkgray] (name6) at (axis cs:-0.6,-1.9) {$I_5$};

                \draw[] (axis cs:0.415,1.35) -- (axis cs:0.21,1.73);
                \draw[] (axis cs:-0.55,-1.85) -- (axis cs:-0.21,-1.73);

            \end{axis} 
            \end{tikzpicture} 

        \caption{Full covering of the $y$-axis}
        \label{fig:ex-cell}
    \end{subfigure}
    \caption{\footnotesize{\emph{Left:} Illustration of the polynomial zeros from the defining polynomials of $\varphi$. \emph{Right:} A full covering of the $y$-axis which is achieved after six sampling steps. Every cross-hatched bar represents an open interval. The two horizontal black lines correspond to the sections $I_5$ and $I_6$. The coloured area belongs to the cell induced by the interval $I_1$.}}
\end{figure}%

\subsection{Original CAlC method}

\paragraph{Initial assignment}
The CAlC method starts with the empty sample $s=()$. Before making a choice for $s_1$, every constraint in $\varphi$ is substituted by the current partial sample $s$ and the univariate ones are analysed. However, without assigning any variable, no constraint becomes univariate, i.e. no UNSAT intervals can be drawn from the constraints. Therefore, any real number is suitable for being assigned to $x$. For simplicity, we choose $0$ for $x$. The updated sample $s=(0)$ has not yet full dimension, so we continue recursively and set $y$ to $0$, obtaining $s=(0,0)$.

\paragraph{UNSAT intervals from constraints}
This time, substitution $\varphi(0,0,z)$ of $s$ produces univariate constraints that yield the UNSAT intervals given in \Cref{tab:ex1}. We use `$\approx{\,\,\cdot\,\,}$' to indicate rounded numbers.
\begin{table}[t]
    \centering
    \begin{minipage}{0.99\textwidth}
        \centering
        \captionof{table}{Partially evaluated constraints and the intervals of unsatisfaction.}
        \label{tab:ex1}
        \begin{tabular}{ccc}
        \toprule
        Constraint & Evaluated & Unsatisfied Intervals \\ \midrule
        $c_1$ & $z \geq 0$      & $(-\infty,0)$ \\
        $c_2$ & $z^2 < 3$      & $(-\infty,-\approx{1.73}), [-\approx{1.73},-\approx{1.73}], [\approx{1.73},\approx{1.73}], (\approx{1.73},\infty$) \\
        $c_3$ & $z^2 > 9.75$      & $[-\approx{3.12},-\approx{3.12}], (-\approx{3.12},\approx{3.12}), [\approx{3.12},\approx{3.12}].$ \\ \bottomrule
        \end{tabular}
    \end{minipage}
\end{table}
Alongside with each UNSAT interval $I$ at $s$, the CAlC algorithm stores a set of polynomials and a sample $(s,s')$ where $s' \in I$ to implicitly represent the corresponding UNSAT cell that is witnessed by $I$. E.g. alongside the interval $(-\infty,0)$ from $c_1$ we would also store $\{ p_1 \}$ and $(0,0,-1)$. A more technical view on this example is provided in \cite{ba}.

\paragraph{Projection}
The UNSAT intervals obtained from $c_2$ and $c_3$ cover the real line. Hence, no satisfying extension of $s$ onto the $z$-dimension is possible. Not all intervals from $c_2$ and $c_3$ are needed for a full covering, so a non-redundant selection is made, e.g. we can choose the set $\I := \{ (-\infty,-\approx{1.73}), (-\approx{3.12},\approx{3.12}), (\approx{1.73},\infty) \}$.

Now the method tracks back to the $y$-dimension and projection polynomials characterizing the underlying cell of $\I$ is calculated. The characterization consists of resultants, discriminants and coefficients; we do not discuss the details here. In this case, we obtain the projection polynomials 
\begin{align*}
    p_4 &:= \disc(p_2) = y^2 + x^2 - 3,\\
    p_5 &:= \disc(p_3) = y^2 - 5y + x^2 - 9.75,\\
    p_6 &:= \res(p_2,p_3) = y + 1.35.
\end{align*}%

\paragraph{Interval deduction for $\mathbf{y}$}
We compute an interval for $y$ from the computed projection polynomials; if we would set $y$ to any value within this interval, we would obtain the same conflict (i.e. the same cells would cover all values for $z$) as $(0,0)$.


At first, the roots of $p_4(0,y),p_5(0,y)$, and $p_6(0,y)$ are calculated based on the lower-dimensional sample $s=(0)$. This results in the roots
\begin{displaymath}
    Z = \{ -\infty, \underbrace{-\approx{1.73}}_{p_4}, \underbrace{-1.5\vphantom{\underline{5}}}_{p_5}, \underbrace{-1.35\vphantom{\underline{5}}}_{p_6}, \underbrace{\approx{1.73}}_{p_4}, \underbrace{6.5\vphantom{\underline{5}}}_{p_5}, \infty \}.
\end{displaymath}
The closest zeros around the current assignment $0$ for $y$ are $-1.35$ and $\approx{1.73}$. In between, the original polynomials do not show significant changes in their behaviour, i.e. the number and order of their zeros remains the same. Thus, we obtain the interval $I_1 := (-1.35,\approx{1.73})$ together with the polynomials $\{p_6,p_4\}$ and the sample $(0,0)$ defining the corresponding UNSAT cell. A visualization of the current situation is given in \Cref{fig:ex-cell}. The cyan area belongs to the cell induced by the interval $I_1$. It is bounded by the zeros of $p_4$ and $p_6$. The expansion of the cell in $x$-dimension is limited by the common zeros of $p_4$ and $p_6$. The cylinder over the cell from \phantom{withgen}

\begin{minipage}{\linewidth}
    \centering
    \hspace*{-0.6cm}
    \begin{minipage}{0.55\textwidth}
    \centering
    \setkomafont{caption}{\sffamily\small}
    \setkomafont{captionlabel}{\usekomafont{caption}}
    \captionsetup{labelfont=bf,labelsep=newline,format=plain}
    \captionof{table}{Partially evaluated constraints and the intervals of unsatisfaction for the modified method.\\}
    \label{tab:ex2}
    \begin{tabular}{ccc}
        \toprule
        Constraint & Evaluated & Unsatisfied Intervals \\ \midrule
        $c_1$ & $z \geq 0$      & $(-\infty,0)$ \\
        $c_2$ & $z^2 < 3$      & $(-\infty,-\approx{1.73}],\![\approx{1.73},\infty$) \\
        $c_3$ & $z^2 > 9.75$      & $[-\approx{3.12},\approx{3.12}].$ \\ \bottomrule
    \end{tabular}

    \phantom{W}

    \phantom{w}
    \end{minipage}%
    \hspace*{0.35cm}%
    \begin{minipage}{0.42\textwidth}
    \centering
    \setkomafont{caption}{\sffamily\small}
    \setkomafont{captionlabel}{\usekomafont{caption}}
    \captionsetup{labelfont=bf,labelsep=colon,format=plain}
    \begin{tikzpicture}[sibling distance=10em,
        every node/.style = {align=center,sibling distance=1cm,level distance=1cm},
        level/.style={sibling distance=2.1cm,level distance=1cm},
        level 2/.style={dashed,level distance=0.8cm}]]
        \node {$[-1.35,\approx{1.73}]$}
            child { node {$(-\infty,-\approx{1.73}]$} 
            child { node { $c_2$}}
            }
            child { node {$[-\approx{3.12},\approx{3.12}]$}
            child { node { $c_3$}}
            }
            child { node {$[\approx{1.73},\infty)$}
            child { node { $c_2$}}
            };
    \end{tikzpicture}
    
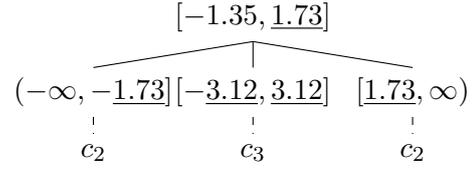
\captionof{figure}{Ancestor constraint tree of the interval $I_1$. Dashed edges do not belong to the tree, but draw the connection to input constraints.}
    \label{fig:ex-tree}
    \end{minipage}
\end{minipage}

\noindent \Cref{fig:ex-cell} is covered by the UNSAT cells induced by $p_2$ and $p_3$. That is, any two-dimensional sample from the cell cannot be extended for a satisfying sample.

\paragraph{Iterative procedure}
$I_1$ does not cover the whole $y$-axis yet, so we continue sampling $s=(0,a)$ for some $a \in \RR \setminus I_1$. When using the sequence $a_1 := 0, a_2 := -1.5, a_3 := -2, a_4 := 2, a_5 := -\approx{1.73}, a_6 := \approx{1.73}$ to choose $y$, we obtain the intervals depicted in \Cref{fig:ex-cell}. The interval $I_j$ is deduced after sampling $a_j$. Now the $y$-axis is covered, and projection steps take place again.

\paragraph{Interval deduction for $\mathbf{x}$}
Based on $I_1, \dots, I_6$ a characterization is computed, i.e. based on intervals that do not stem from constraints directly. In doing so, we deduce the interval $(-\approx{1.09},\approx{1.09})$ for the $x$-dimension.
If one would continue updating $s_1$, a full covering of the $x$-axis can be found. Hence, when backtracking to $s=()$, no value for assigning $x$ is left uncovered and overall UNSAT is returned. The unsatisfiability of the constraint set is detected.

\subsection{Modified CAlC method}
Now we apply the modified CAlC method using \Cref{thm:calc-ext} to $\varphi$. We refer to the cell flags by $\d$.

\paragraph{UNSAT intervals from constraints}
When choosing the sample $s=(0,0)$, all three constraints can be partially evaluated resulting in the UNSAT intervals in \Cref{tab:ex2}.

Note that sections and sectors seem to be merged because open interval bounds can be closed up for strict constraints. All intervals except the first have the flag $\d=\True$ as $c_2$ and $c_3$ are strict.

\paragraph{Interval deduction for $\mathbf{y}$}
The intervals from $c_2$ and $c_3$ form a non-redundant covering for the $z$-dimension and a characterization of the conflict is computed. Again, the projection yields the polynomials $p_4,p_5$, and $p_6$ with the same zeros $Z$.
The closest ones still are $-1.35$ and $\approx{1.73}$. We deduce the interval $I_1 := [-1.35,\approx{1.73}]$ because the tree of coverings above $I_1$ consists of three leaf intervals (two from $c_2$ and one from $c_3$) all with $\d=\True{}$. The conjunction of the three interval flags results in $\d=\True{}$ for $I_1$. The tree is depicted in \Cref{fig:ex-tree}.
The possibility of reassigning $s_2$ is reduced by the values $-1.35$ and $\approx{1.73}$.

If one samples the $y$-axis with $a_1, a_2, a_3$, and $a_4$, a covering of four instead of six intervals is derived: $I_1,I_3$, and $I_4$ have closed finite endpoints, $I_2$ still does not include its bounds because $c_1$ is not strict. The sections $I_5$ and $I_6$ do not appear at all, because there is not need to sample $\pm\approx{1.73}$ as they are already covered by the closed endpoints of $I_1$ and $I_3$.

\paragraph{Interval deduction for $\mathbf{x}$}
Characterization steps are performed yielding the same polynomials $p_4,p_5$, and $p_6$ as before. In particular, the two intervals $I_5$ and $I_6$ from \Cref{fig:ex-cell} are redundant now.
The first interval covering a part of the $x$-axis is $(-\approx{1.09},\approx{1.09})$. Because the flag of $I_2$ is \False{}, the conjunction of the corresponding flags of $I_1,I_2,I_3$, and $I_4$ is \False{}, too.

At this point, the theorem cannot be applied, because the cell induced by $I_2$ does not correspond to its closure.
In case $c_1$ would be strict as well, the deduced interval $I_2$ would be closed up, too. That way, also the interval for $x$ could be closed up.
The example points out that \Cref{thm:calc-ext} makes strong assumptions on the covering. Therefore, the number of closed intervals tends to decrease by every level.

\fi

\end{document}